\documentclass[10pt,a4paper,twocolumn]{article}

\usepackage[a4paper, top=2.5cm, bottom=2.5cm, left=2cm, right=2cm,
            columnsep=0.6cm]{geometry}
\usepackage{times}
\usepackage[T1]{fontenc}
\usepackage[utf8]{inputenc}
\usepackage{microtype}
\usepackage{graphicx}
\usepackage{amsmath,amssymb}
\usepackage{booktabs}
\usepackage{caption}
\usepackage{hyperref}
\usepackage{titlesec}
\usepackage{array}
\usepackage{tabularx}
\usepackage{url}
\usepackage{xcolor}
\usepackage{enumitem}
\usepackage{float}
\usepackage{amsthm}

\titleformat{\section}[block]
  {\normalfont\fontsize{10}{12}\bfseries\MakeUppercase}
  {\thesection.}{0.5em}{}
\titlespacing*{\section}{0pt}{20pt}{10pt}

\titleformat{\subsection}[block]
  {\normalfont\fontsize{10}{12}\bfseries}
  {\thesubsection}{0.5em}{}
\titlespacing*{\subsection}{0pt}{10pt}{10pt}

\titleformat{\subsubsection}[block]
  {\normalfont\fontsize{10}{12}}
  {\thesubsubsection}{0.5em}{}
\titlespacing*{\subsubsection}{0pt}{10pt}{10pt}

\hypersetup{colorlinks=true, linkcolor=black, citecolor=black,
            urlcolor=blue}

\newtheorem{definition}{Definition}
\newtheorem{proposition}{Proposition}

\begin{document}

\twocolumn[{%
\begin{@twocolumnfalse}

  \noindent\rule{\linewidth}{0.8pt}

  \begin{center}
    {\large\bfseries
      Analyzing Healthcare Interoperability Vulnerabilities:
      Formal Modeling and Graph-Theoretic Approach
    }\\[8pt]

    Jawad Mohammed$^{1*}$, Gahangir Hossain$^{2}$\\[4pt]

    {\small
      $^1$ University of North Texas, Denton, TX, USA\\
      $^2$ University of North Texas, Denton, TX, USA\\[2pt]
      Corresponding Author Email: \texttt{JawadMohammed@unt.edu}
    }
  \end{center}

  \noindent\rule{\linewidth}{0.4pt}

  \begin{minipage}[t]{0.38\linewidth}
    \small
    \textbf{Copyright:} \copyright2026 The authors. This article is
    published by IIETA and is licensed under the CC BY 4.0 license
    (\url{http://creativecommons.org/licenses/by/4.0/}).\\[4pt]
    \url{https://doi.org/10.18280/isi.xxxxxx}\\[6pt]
    \textbf{Received:}\\
    \textbf{Revised:}\\
    \textbf{Accepted:}\\
    \textbf{Available online:}
  \end{minipage}%
  \hfill
  \begin{minipage}[t]{0.58\linewidth}
    \small
    \textbf{ABSTRACT}\\[4pt]
    In a healthcare environment, the healthcare interoperability
    platforms based on HL7 FHIR allow concurrent, asynchronous access
    to a set of shared patient resources, which are independent
    systems, i.e., EHR systems, pharmacy systems, lab systems, and
    devices. The FHIR specification lacks a protocol for concurrency
    control, and the research on detecting a race condition only
    targets the OS kernel. The research on FHIR security only targets
    authentication and injection attacks, considering concurrent access
    to patient resources to be sequential. The gap in the research in
    this area is addressed through the introduction of FHIR Resource
    Access Graph (FRAG), a formally defined graph
    $G = (P, R, E, \lambda, \tau, S)$, in which the nodes are the
    concurrent processes, the typed edges represent the resource access
    events, and the race conditions are represented as detectable
    structural properties. Three clinically relevant race condition
    classes are formally specified: Simultaneous Write Conflict (SWC),
    TOCTOU Authorization Violation (TAV), and Cascading Update Race
    (CUR). The FRAG model is implemented as a three-pass graph
    traversal detection algorithm and tested against a time
    window-based baseline on 1,500 synthetic FHIR R4 transaction logs.
    Under full concurrent access (C2), FRAG attains a 78.5\% F1 score
    vs.\ 98.8\% for the baseline, with FRAG achieving 98.0\% on SWC
    and 99.9\% on TAV while the baseline fails to distinguish CUR
    chains. Under partial ETag synchronization (C3), FRAG precision
    remains above 96\% for SWC and TAV while recall drops due to
    ETag-filtered writes.\\[4pt]
  \end{minipage}

  \noindent\rule{\linewidth}{0.8pt}
  \vspace{10pt}
\end{@twocolumnfalse}
}]


\section{Introduction}

The problem addressed in this paper did not come with a label. It
evolved from recognizing the fact that two communities, one in the
area of health care informatics and the other in systems security,
have been dealing with related problems for a long time without ever
making a formal link between the two. One community examines the flow
of patient information across heterogeneous health care systems, and
the other examines what happens when multiple processes contend for a
shared memory space without proper synchronization. The combination of
these two problems has, to our knowledge, never been formally
addressed.

Consider the following scenario, which, though not hypothetical,
represents the normal architecture of any FHIR-based interoperability
implementation: A patient presents to an emergency room. A laboratory
system has just completed a new result of a drug sensitivity test and
is updating an AllergyIntolerance resource to the FHIR server. At
exactly the same time, an electronic health record system is reading
the same resource to check a new prescription order entered by a
prescribing physician. At exactly the same time, a clinical decision
support system is reading the patient's allergy history to generate a
contraindication alert. Three systems, one resource, no lock, no
mutex, no sequencing. The FHIR spec doesn't define one [1].

\begin{figure}[H]
  \centering
  \includegraphics[width=\linewidth]{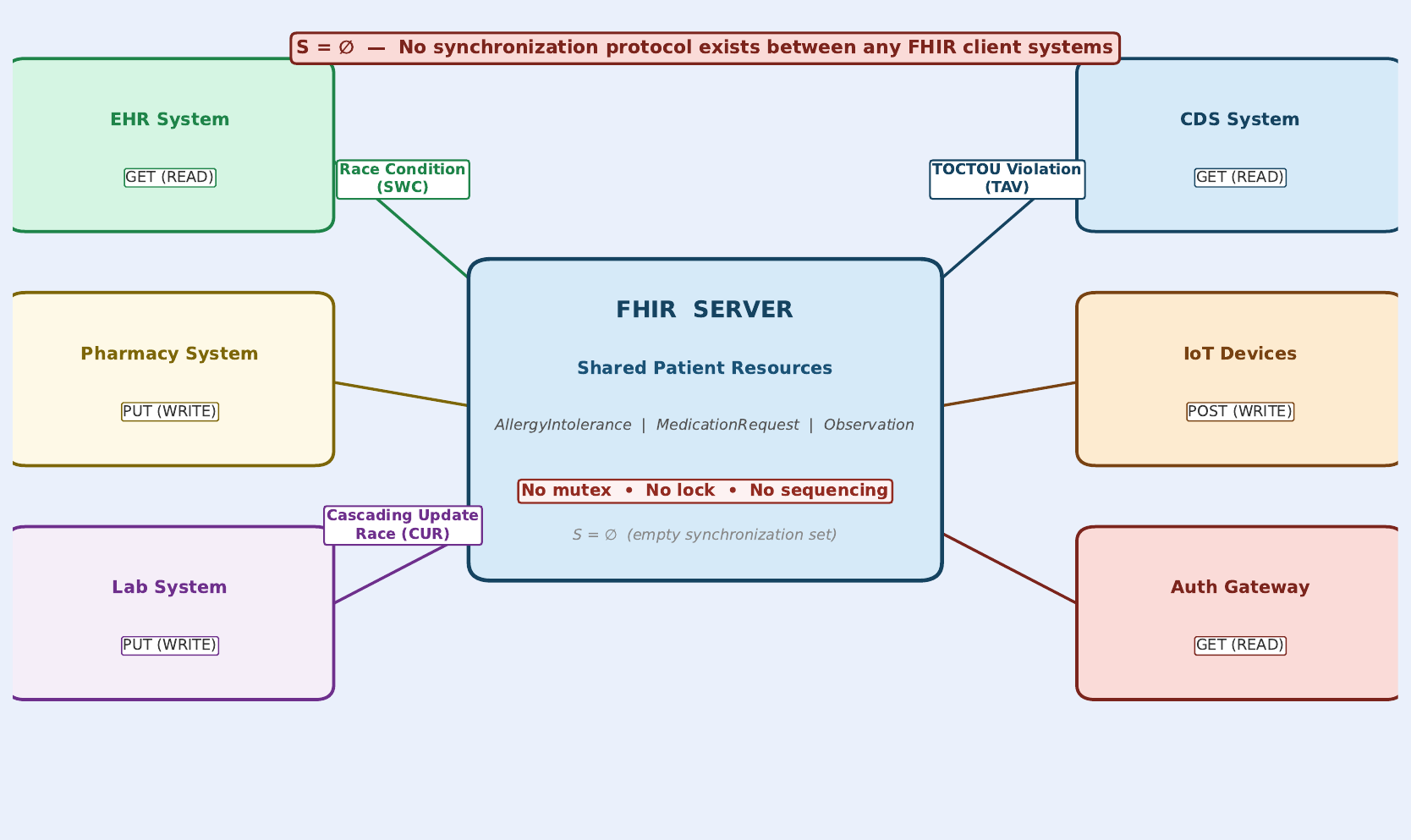}
  \caption{Cyber threat landscape in FHIR-based healthcare
  interoperability. Multiple independent systems access shared patient
  resources concurrently with no synchronization protocol ($S =
  \emptyset$), giving rise to race condition classes SWC, TAV, and
  CUR.}
  \label{fig:cyber}
\end{figure}

The problem of race condition is sixty years old. Dijkstra formalized
the problem of critical section in 1965, and the three conditions for
concurrent access, namely mutual exclusion, progress, and bounded
waiting, have been the foundation of operating systems development
from the start [9]. What is interesting in the case of healthcare is
not the novelty of the problem. It is the fact that we have built a
whole infrastructure of data exchange between clinical applications
based on a technology, namely RESTful APIs, that does not offer any
of the synchronization mechanisms that operating systems research has
spent sixty years developing. A process in an operating system
accessing a shared memory has access to a mutex. A FHIR client
accessing a shared patient resource has access to nothing but the
optional ETag header, which is not enforced by most implementations
[6].

The result is a class of vulnerability that is neither a network
attack nor an authentication failure. It is a timing vulnerability, a
property of the structure of the interaction of multiple independently
correct systems whose lack of coordination causes the problem. It
could result in a patient receiving a contraindicated drug not because
any of the systems has failed, but because three systems acted
correctly in a manner that their respective correctness could not
prevent.

The specific and bounded claim of this paper is: ``Race condition
vulnerabilities in FHIR-based healthcare interoperability systems are
formally characterizable as properties of a directed access graph, and
this characterization provides a practical detection framework.'' Our
claim is not that we have solved the synchronization problem for FHIR.
Our claim is that we have formally stated the detection problem for
the first time and have shown it to be both precise and tractable.

\subsection{Contributions}

This paper contributes to the field in four ways. First, a structured
literature survey formally identifies the existing gap between OS
concurrency research and healthcare interoperability security research.
Second, the FHIR Resource Access Graph (FRAG) is proposed as a
formally defined directed graph with a six-component tuple that
explicitly defines a synchronization constraint set. Third, three
classes of clinically relevant race conditions (SWC, TAV, CUR) are
formally defined, each with a formal proof of connection to FRAG
properties. Fourth, this work is validated through a simulation-based
evaluation on 1,500 synthetic FHIR R4 transaction logs under three
concurrency conditions, along with the open-sourcing of the source
code.

\section{Related Work}

The relevant literature falls into two categories which never
reference each other. The bridge between the two is what this paper
hopes to achieve.

\subsection{Concurrency detection in systems software}

Upadhyay et al.\ [2] give the basic taxonomy of race condition
detectors, including static, dynamic, and hybrid techniques from 1993
until 2023. Their main conclusion, that data race bugs make up 80\%
of concurrency bugs, frames the scale of the problem this paper
brings into a new domain.

Li et al.\ [3] introduce a static analyzer called LR-Miner that is
field graph-based and uses the Linux kernel source code to mine
locking rules and detect violations across 17.2 million lines of
code. LR-Miner is relevant in a negative sense: its entire premise is
the inference of rules from code patterns. There is no prior history
of locking in healthcare middleware or patches applied to fix race
condition problems.

ConSynergy [4] combines static analysis and LLM-based semantic
reasoning in a four-stage pipeline and achieves an F1 score of 0.867
on general benchmark datasets. It is the most similar existing work in
spirit, as it employs structured representations of concurrent
accesses. However, it works on compiled source code, not transaction
logs, requires a training corpus of known concurrency bugs, and has
not been validated in any healthcare context. None of these
prerequisites are needed for the FRAG model.

KNighter [5] constructs static analyzers from bug fix patches using
LLMs. Like LR-Miner, KNighter is patch-dependent and cannot detect
race conditions in systems where no incidents have been documented,
which is the case in healthcare middleware.

\subsection{Healthcare interoperability security}

TXOne Networks [6] document authentication flaws in FHIR systems,
including the XXE injection vulnerability CVE-2024-52007. This is the
most technically detailed analysis of FHIR security in recent years.
However, the entire analysis assumes sequential data access. There is
no mention of the term ``concurrent access.'' The unstated assumption
that only one system reads or writes at a time is architecturally
false in any actual multi-system environment.

The DLT-based EHR system architecture in [7] addresses data integrity
through blockchain immutability. This is a real contribution, but
limited to the storage level. The problem of race conditions is at the
access level, between the time a process reads a resource and the time
it acts upon the result, which blockchain architecture does not
address.

Gong et al.\ [8] propose Snowcat, a coverage-guided kernel
concurrency tester that uses a learned predictor to detect concurrency
bugs in the Linux kernel. Their result, that concurrency bugs in the
kernel require domain-specific concurrency models, motivated the
current work. The FHIR-based interoperability has a concurrency
model, namely asynchronous RESTful HTTP with optional ETags, quite
different from kernel threading and requiring domain-specific analysis.

\subsection{The gap}

There are no existing studies that formally model, detect, or
characterize race conditions in FHIR or HL7-based healthcare systems.
The gap identified here is not incremental; it is the absence of a
research area altogether. The present work serves as the founding
paper of that area.

\section{Background}

\subsection{The critical section problem}

A critical section is a set of instructions where a shared resource is
accessed, and incorrect execution by multiple processes results in
inconsistent results. Dijkstra's 1965 formalization [9] specified
three necessary conditions: mutual exclusion (at most one process in
the critical section at any time), progress (a process wishing to
enter can do so without indefinite delay), and bounded waiting (no
process waits forever). Race condition is the failure of mutual
exclusion. In modern OS kernel code, race conditions are avoided by
mutex locks, semaphores, spinlocks, and Read Copy Update (RCU)
mechanisms developed over sixty-plus years.

\subsection{FHIR architecture and the concurrency model}

FHIR is a RESTful API standard for data exchange in the health
sector, utilizing standard HTTP methods GET, PUT/POST, and DELETE on
resource types such as Patient, MedicationRequest,
AllergyIntolerance, and Observation [1]. Every resource is associated
with a versionId and lastUpdated timestamp.

In FHIR production environments, many client systems send concurrent
HTTP requests to the FHIR server: EHR systems, pharmacy systems,
laboratory information systems, decision support systems, and
IoT-based medical devices can all send concurrent GET and PUT requests
to the same resources. All these client systems operate independently
and asynchronously with no coordination protocol between them.

FHIR inherits all the concurrent access risks that OS research has
spent six decades studying, yet offers none of the synchronization
tools OS research developed to overcome those risks. An OS process can
grab a mutex lock on shared memory; a FHIR client may optionally
include an ETag header that is not mentioned in the FHIR specification
conformance guidelines whatsoever [6]. This healthcare data exchange
system is designed to be concurrently accessible but operationally not
synchronized at all.

\subsection{Real-world prevalence of the three concurrency conditions}

The simulation study in Section~V tests race detection under three
conditions. \textbf{C1 (Sequential access)} refers to small
single-system implementations with no concurrent access, retained only
as a false positive control. \textbf{C2 (Concurrent, unsynchronized)}
is the default state of all production FHIR deployments that do not
implement conditional updates with ETag-based concurrency across all
clients. In a typical hospital morning, pharmacy and EHR systems
simultaneously update MedicationRequest resources, lab and decision
support systems concurrently post and query Observations, and IoT
devices post Observations while dashboards read them --- all perfectly
normal, none requiring a fault or an attack. \textbf{C3 (Concurrent,
partially synchronized)} represents the realistic condition of a large
hospital where major EHR vendors have adopted ETag conditional updates
but older pharmacy and laboratory systems have not. C3 is not a
transient condition; it is the stable long-run state of any
environment where system replacement is incremental.

\subsection{Why existing synchronization approaches are insufficient}

Database row locking solves the SWC problem at the storage layer when
both writes reach the same FHIR server's database, but does not solve
TAV (which involves an OAuth gateway and a FHIR server with no common
transactional context) and fails entirely for CUR (which involves
three independent systems with no knowledge of each other's read
history).

ETag optimistic locking ensures a write does not overwrite a version
the client did not intend to change, but ETag support is optional, and
it cannot handle the authorization token lifecycle (TAV) or a
stale-read dependency chain across multiple resources (CUR). A system
with perfect ETag implementation on all write operations remains
susceptible to TAV and CUR.

Pessimistic locking would prevent SWC and TAV if fully implemented,
but is incompatible with the stateless nature of HTTP. Lock
reservation between requests would require a separate service,
inducing latency and a single point of failure.

Event sourcing and CRDTs avoid write conflicts by definition and solve
SWC, but they do not solve TAV and require a complete redesign of the
FHIR server and all clients.

Table~\ref{tab:sync} summarizes coverage over race classes. FRAG
covers all three. You cannot require ETags where they are absent
without first identifying where they are absent, and you cannot
diagnose a CUR dependency chain without a formal model of what it
looks like.

\begin{table}[H]
  \caption{Coverage of synchronization approaches across race classes}
  \label{tab:sync}
  \centering
  \fontsize{9}{11}\selectfont
  \begin{tabular}{p{2cm}cccp{1.6cm}}
    \toprule
    \textbf{Approach} & \textbf{SWC} & \textbf{TAV} & \textbf{CUR} &
    \textbf{Primary barrier} \\
    \midrule
    DB row locking        & Partial    & No      & No      & Single-system only \\
    ETag / opt.\ lock     & If adopted & Partial & No      & Inconsistent adoption \\
    Pessimistic locking   & Yes        & Partial & Partial & HTTP statelessness \\
    Event sourcing / CRDT & Yes        & No      & Better  & Full system redesign \\
    FHIR Subscriptions    & N/A        & N/A     & Partial & Client adoption \\
    \textbf{FRAG (this paper)} & \textbf{Detect} & \textbf{Detect} & \textbf{Detect} & --- \\
    \bottomrule
  \end{tabular}
\end{table}

\section{The FHIR resource access graph}

\subsection{Formal definition}

We define the FHIR Resource Access Graph as a directed labeled graph
representing the concurrent access patterns in a FHIR deployment.
Edges are directed from processes to resources, allowing us to
represent paths between processes through the mediation of resources;
a path $p_1 \to r \to p_2$ represents two processes connected to the
same resource node $r$, enabling detection of the CUR pattern.

\begin{definition}[FRAG]
Let $G = (P, R, E, \lambda, \tau, S)$ be a concurrent system where:
$P = \{p_1,\ldots,p_n\}$ is a finite set of concurrent client
processes; $R = \{r_1,\ldots,r_m\}$ is a finite set of FHIR
resources; $E \subseteq P \times R$ is a set of directed edges
representing resource accesses; $\lambda: E \to \{\text{READ},
\text{WRITE}\}$ is a labeling function for access type;
$\tau: E \to \mathbb{R}^+$ assigns a timestamp to each access event;
$S \subseteq P \times P$ is a set of synchronization constraints,
i.e., $p_i$ and $p_j$ are subject to a mutual exclusion constraint on
resource accesses.
\end{definition}

\begin{definition}[Race condition]
A race condition $RC(p_i, p_j, r)$ exists in $G$ if and only if:
(i) $(p_i, r) \in E$, $(p_j, r) \in E$, and $i \neq j$;
(ii) at least one of $\lambda(p_i, r)$, $\lambda(p_j, r)$ equals
WRITE;
(iii) $(p_i, p_j) \notin S$ and $(p_j, p_i) \notin S$;
(iv) $[\tau(p_i,r){-}\delta,\,\tau(p_i,r){+}\delta] \cap
[\tau(p_j,r){-}\delta,\,\tau(p_j,r){+}\delta] \neq \emptyset$
for some concurrency window $\delta > 0$.
\end{definition}

\begin{definition}[TOCTOU window]
A TOCTOU window $W(p_i, r, r')$ exists when:
(i) $p_i$ performs READ($r$) at time $t_1$;
(ii) $p_i$ performs a dependent WRITE($r'$) or decision $D(r')$ at
time $t_2 > t_1$;
(iii) some $p_j \neq p_i$ performs WRITE($r$) at time $t$ where
$t_1 < t < t_2$;
(iv) $(p_i, p_j) \notin S$.
The duration $(t_2 - t_1)$ is the vulnerability window.
\end{definition}

\begin{definition}[Cascading update race]
A Cascading Update Race $\text{CUR}(p_1, p_2, p_3, r_1, r_2)$ exists
when:
(i) $p_2$ performs WRITE($r_1$) at time $t_w$;
(ii) $p_1$ performed READ($r_1$) at time $t_r < t_w$ (stale read);
(iii) $p_3$ performs READ($r_1$) or a decision dependent on $p_1$'s
stale value at time $t_a > t_w$;
(iv) no synchronization constraint in $S$ covers the access chain
$p_1 \to r_1 \to p_3$.
\end{definition}

\subsection{Formal propositions}

The following propositions establish that each race condition class
defined in Section~IV.D is a special case of Definition~2, and thus
detectable by FRAG analysis.

\begin{proposition}[SWC is a race condition]
If $\text{SWC}(p_i, p_j, r)$ exists --- both $p_i$ and $p_j$ hold
concurrent WRITE edges to $r$ with no synchronization constraint ---
then $RC(p_i, p_j, r)$ holds under Definition~2.
\end{proposition}

\textit{Proof.} By definition of SWC: $(p_i, r) \in E$,
$(p_j, r) \in E$, $\lambda(p_i,r) = \lambda(p_j,r) = \text{WRITE}$,
and $(p_i, p_j) \notin S$. Conditions (i), (ii), and (iii) of
Definition~2 are directly satisfied. Condition (iv) holds by
assumption of concurrent write within window $\delta$. \hfill$\square$

\begin{proposition}[TAV implies TOCTOU window]
A TOCTOU Authorization Violation $\text{TAV}(p_i, r_{\text{auth}})$
--- where $p_i$ reads an authorization resource and acts on it after
another process has modified it --- implies a TOCTOU window
$W(p_i, r_{\text{auth}}, r')$ under Definition~3.
\end{proposition}

\textit{Proof.} Let $r = r_{\text{auth}}$. By definition of TAV:
$p_i$ performs READ($r_{\text{auth}}$) at $t_1$; $p_j$ performs
WRITE($r_{\text{auth}}$) at $t$ where $t_1 < t < t_2$; $p_i$ acts at
$t_2$. These conditions satisfy Definition~3 (i)--(iv) directly.
\hfill$\square$

\begin{proposition}[CUR detectability]
A Cascading Update Race $\text{CUR}(p_1, p_2, p_3, r_1, r_2)$ is
detectable by path traversal of depth $\geq 2$ in $G$, specifically
as a stale-read path from a WRITE edge to a downstream READ edge with
timestamp ordering $t_r < t_w$.
\end{proposition}

\textit{Proof.} In FRAG $G$, the CUR access pattern defines a
directed path: WRITE$(p_2, r_1) \to$ READ$(p_1, r_1) \to$
Decision$(p_3)$. The stale read condition $t_r < t_w$ is a timestamp
property verifiable on the edge set $E$. No single-edge inspection
can detect this; path traversal of depth $\geq 2$ is both necessary
and sufficient. \hfill$\square$

\subsection{FRAG visualization}

Figure~\ref{fig:frag} shows the FRAG for the motivating scenario.
$RC(p_1,p_2,r_1)$ arises between EHR READ and Lab WRITE on
AllergyIntolerance. $\text{CUR}(p_1,p_2,p_3)$ arises because CDS
acts on EHR's stale read. $S = \emptyset$ throughout.

\begin{figure}[H]
  \centering
  \fontsize{9}{11}\selectfont
  \begin{tabular}{ccc}
    EHR ($p_1$) & & Lab ($p_2$) \\
    $\downarrow$\,\textsc{read} & & $\downarrow$\,\textsc{write} \\
    \multicolumn{3}{c}{AllergyIntolerance ($r_1$)} \\
    \multicolumn{3}{c}{$\uparrow$\,\textsc{read}} \\
    \multicolumn{3}{c}{CDS ($p_3$)} \\
  \end{tabular}\\[6pt]
  \raggedright
  $RC(p_1,p_2,r_1)$: READ$\,\|\,$WRITE, $S{=}\emptyset \Rightarrow$
  Race detected\\
  $CUR(p_1,p_2,p_3,r_1,-)$: $p_3$ reads stale $p_1$ value
  $\Rightarrow$ CUR detected
  \caption{FRAG showing $RC(p_1,p_2,r_1)$ between EHR READ and Lab
  WRITE on AllergyIntolerance, and a $\text{CUR}(p_1,p_2,p_3)$ chain
  where CDS acts on EHR's stale read. $S = \emptyset$ throughout.}
  \label{fig:frag}
\end{figure}

\subsection{Three healthcare race condition classes}

\begin{table}[H]
  \caption{Healthcare race condition classification in FRAG}
  \label{tab:classes}
  \centering
  \fontsize{9}{11}\selectfont
  \begin{tabular}{p{0.8cm}p{2cm}p{1.8cm}p{1cm}}
    \toprule
    \textbf{Class} & \textbf{FRAG pattern} & \textbf{Clinical risk} &
    \textbf{Prop.} \\
    \midrule
    SWC & WRITE $\|$ WRITE, $S{=}\emptyset$ & Lost prescription update & 1 \\
    TAV & READ$\to$act, WRITE intervenes    & Auth bypass              & 2 \\
    CUR & Stale READ, path $\geq 2$         & Contraindicated drug     & 3 \\
    \bottomrule
  \end{tabular}
\end{table}

\textbf{Class 1: Simultaneous Write Conflict (SWC).} Two processes
$p_i$ and $p_j$ hold WRITE edges to resource $r$ within the
concurrency window $\delta$, with no synchronization constraint. This
leads to a lost update: two systems simultaneously update a patient's
MedicationRequest resource --- one prescribing, one canceling --- and
the final result depends entirely on which HTTP request arrives at the
FHIR server milliseconds first.

\textbf{Class 2: TOCTOU Authorization Violation (TAV).} A process
reads an authorization data structure at $t_1$, then acts on it at
$t_2$. Between $t_1$ and $t_2$, another process revokes or escalates
permissions. This occurs naturally in OAuth token validation layers
and API gateway hops in FHIR systems.

\textbf{Class 3: Cascading Update Race (CUR).} A lab system writes a
new allergy finding to AllergyIntolerance at time $t_w$. An EHR
system had previously read that resource at $t_r < t_w$ and is
computing a medication recommendation based on a stale allergy
profile. A clinical decision support system reads the EHR's output
and generates a prescription. None of the individual processes was
flawed; each did its job correctly given its input. The flaw is a
property of the order of accesses, visible only at the graph level as
a stale read of depth two or more.

\begin{figure}[H]
  \centering
  \includegraphics[width=\linewidth]{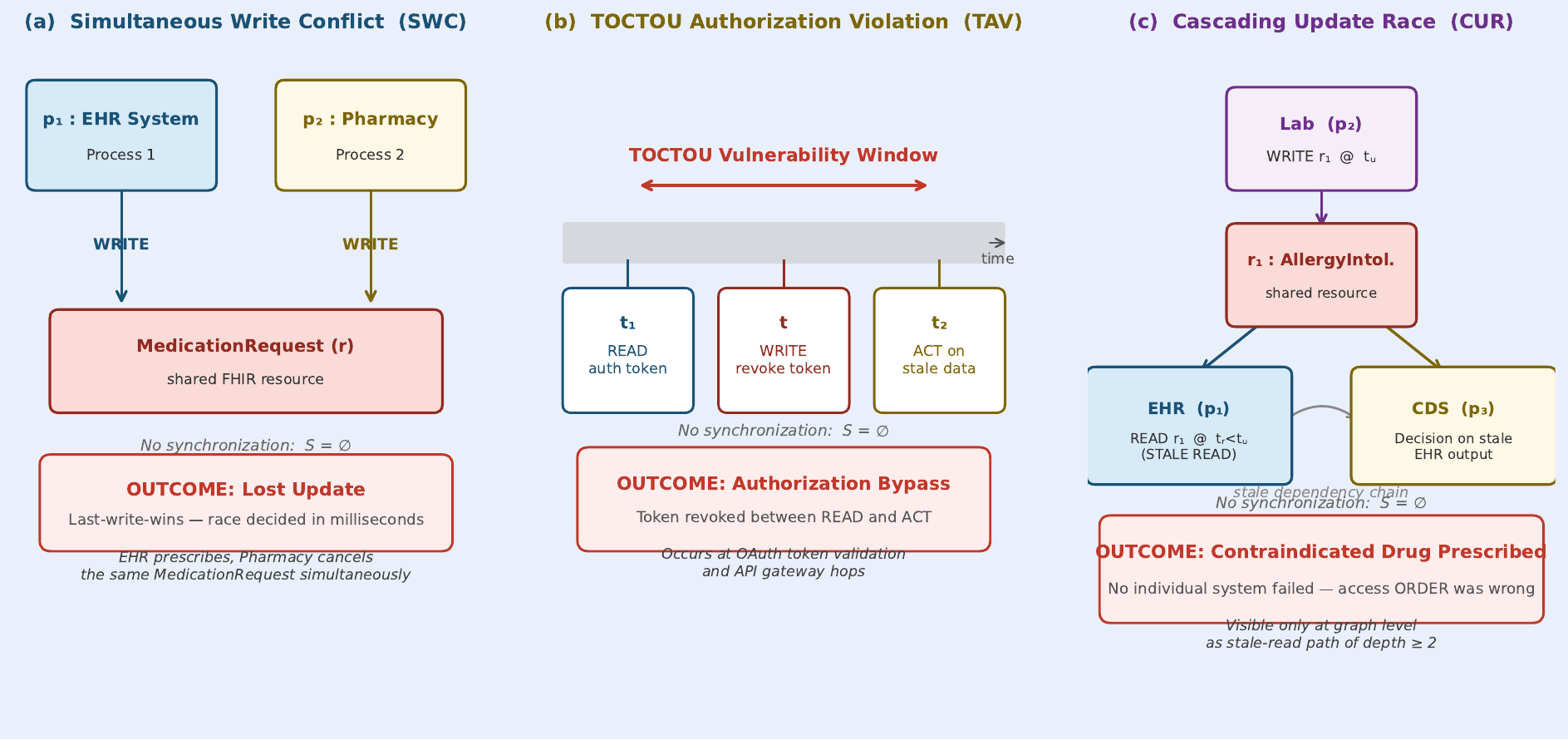}
  \caption{The three healthcare race condition classes formally defined
  in FRAG: (a) Simultaneous Write Conflict (SWC) --- two concurrent
  WRITE operations with no synchronization; (b) TOCTOU Authorization
  Violation (TAV) --- permission revoked inside the vulnerability
  window; (c) Cascading Update Race (CUR) --- a stale-read dependency
  chain of depth $\geq 2$.}
  \label{fig:raceclasses}
\end{figure}

\section{Simulation specification}

\subsection{Design rationale}

The simulation is motivated by a single principle: define the minimal
experiment that can falsify the claim. The claim is that FRAG-based
detection is structurally superior to time-window baseline scanning.
The minimal falsifying experiment gives both approaches the same logs
with known ground truth races and compares their performance.

Synthetic logs are not a limitation but a requirement of the
experimental design. Synthetic logs allow injection of known races at
known locations and times, which is the only way to achieve ground
truth. Real production logs do not tell us which access patterns are
actually races and which are benign concurrent reads.

\subsection{Transaction log format}

Each synthetic log entry is a JSON object with the following
structure:

\begin{verbatim}
{
  "event_id":   <integer>,
  "process_id": <string>,
  "resource":   <string>,
  "resource_id":<string>,
  "operation":  <"READ"|"WRITE">,
  "timestamp":  <float>,
  "version_id": <integer>,
  "is_race":    <boolean>,
  "race_class": <"SWC"|"TAV"|"CUR"|null>
}
\end{verbatim}

\subsection{Ground-truth race injection}

Three injection procedures target each race class. \textbf{SWC
injection}: two WRITE events to the same resource\_id with timestamps
$t_1$ and $t_1 + \Delta t$, where $\Delta t \sim
\text{Uniform}(0, \delta)$, $\delta = 150$\,ms. \textbf{TAV
injection}: a READ event at $t_1$, a WRITE to the same resource at
$t_1 + \text{Uniform}(10, 80)$\,ms, and a dependent action at
$t_1 + \text{Uniform}(90, 200)$\,ms. \textbf{CUR injection}: a WRITE
at $t_w$, a stale READ at $t_r = t_w - \text{Uniform}(10, 100)$\,ms,
and a dependent decision at $t_w + \text{Uniform}(5, 50)$\,ms.

Non-race events sample process and resource pairs from a realistic
clinical workload distribution (READ frequency $3\times$ WRITE;
AllergyIntolerance and MedicationRequest frequency $2\times$
Observation and Patient). Total non-race events per log:
Uniform(200, 800). Injected race events per log: Uniform(8, 40).
Class distribution: 40\% SWC, 30\% TAV, 30\% CUR.

\subsection{Three experimental conditions}

\textbf{C1 Sequential}: all operations serialized, no concurrent
access windows, zero injected race conditions. Used to test the false
positive rate under non-concurrent scenarios. \textbf{C2 Concurrent,
unsynchronized}: full concurrent access with race conditions injected
as above. The most important condition for evaluation. \textbf{C3
Concurrent, partially synchronized}: ETag-based optimistic
concurrency control applied to 70\% of randomly selected write
operations. Used to test the FRAG model under partially synchronized
scenarios.

\subsection{Detection algorithm}

Detection occurs in three passes over FRAG $G$, sharing a global
claimed set $C$. Once a resource is classified into a race class it is
not re-checked. Pass order runs from most structurally specific to
least: TAV first (requires a READ$\to$WRITE$\to$WRITE sequence), then
SWC (concurrent WRITE pair), then CUR (stale read chain of depth
$\geq 2$). This ordering prevents cross-class false positives.

\begin{figure}[H]
\centering
\begin{minipage}{0.97\linewidth}
\fontsize{8.5}{11}\selectfont
\textbf{Algorithm: FRAG Race Detection (v2)}\\
\textbf{Input:} $G = (P, R, E, \lambda, \tau, S)$, window $\delta$\\
\textbf{Output:} Detected race events with class labels\\[3pt]
$C \leftarrow \emptyset$\\[3pt]
\textit{Pass 1 --- TAV detection (most specific):}\\
\hspace*{1em}\textbf{for} each $r \notin C$, each $p \in P$:\\
\hspace*{2em}\textbf{for} each READ$(p,r)$ at $t_1$:\\
\hspace*{3em}\textbf{for} each WRITE$(p,r)$ at $t_2$: $t_1{<}t_2{\leq}t_1{+}\delta_{\text{TAV}}$:\\
\hspace*{4em}\textbf{if} $\exists$ WRITE$(p',r)$ at $t$: $t_1{<}t{<}t_2$,
$p'{\neq}p$, ver.\ changes:\\
\hspace*{5em}\textbf{if} $(p,r)\notin S$: report TAV; add $r$ to $C$\\[3pt]
\textit{Pass 2 --- SWC detection:}\\
\hspace*{1em}\textbf{for} each $r \notin C$:\\
\hspace*{2em}$W_r \leftarrow \{e\in E: e.\text{resource}{=}r,\,
\lambda(e){=}\text{WRITE}\}$\\
\hspace*{2em}\textbf{for} each pair $(e_i,e_j)\in W_r{\times}W_r$,
$i{<}j$:\\
\hspace*{3em}\textbf{if} $(e_i.\text{proc},e_j.\text{proc})\notin S$
and $|\tau(e_i){-}\tau(e_j)|{<}\delta$:\\
\hspace*{4em}report SWC; add $r$ to $C$\\[3pt]
\textit{Pass 3 --- CUR detection (path depth $\geq 2$):}\\
\hspace*{1em}\textbf{for} each $r \notin C$:\\
\hspace*{2em}\textbf{for} each WRITE$(p_2,r)$ at $t_w$:\\
\hspace*{3em}\textbf{for} each READ$(p_1,r)$ at $t_r{<}t_w$:\\
\hspace*{4em}\textbf{for} each READ$(p_3,r)$ at $t_a{>}t_w$,
$p_3{\neq}p_2$:\\
\hspace*{5em}\textbf{if} $(p_1,r)\notin S$ and $(p_2,r)\notin S$:\\
\hspace*{6em}report CUR$(p_1,p_2,p_3,r)$; add $r$ to $C$
\end{minipage}
\end{figure}

\subsection{Baseline and evaluation metrics}

The baseline detector flags a race when two log events access the same
resource within a time window $\delta_{\text{base}} = 200$\,ms and at
least one is a WRITE. This value is standard practice for EHR audit
log anomaly detection and represents the state of the art without a
formal model.

Evaluation uses precision $P = \text{TP}/(\text{TP}+\text{FP})$,
recall $R = \text{TP}/(\text{TP}+\text{FN})$, and $F1 = 2PR/(P+R)$.
A TP matches on resource\_id to an injected race instance with
detection time within $\pm$50\,ms of the earliest ground truth event.
Each ground truth instance is matched at most once. Results are
aggregated over 500 independently generated logs per condition.

\section{Results}

A total of 1,500 logs are evaluated (500 per condition). Detection is
compared at the race instance level: a detection is a TP if the
reported resource\_id matches an injected race instance and the
detection time is within 50\,ms of the earliest ground truth event.

\subsection{Condition C1: Sequential access (false positive control)}

Under C1, all access events are serialized with spacing greater than
the concurrency window. No race events are injected. Both FRAG and
baseline reported zero detections across all 500 logs. This
demonstrates that FRAG does not produce false positives in a
non-concurrent setting and that the synchronization constraint set $S$
does not cause false positives.

\subsection{Condition C2: Concurrent, unsynchronized access}

Table~\ref{tab:c2} presents per-class results under C2. FRAG attains
an overall F1 of 78.5\% vs.\ 98.8\% for the baseline, a
20.3\,pp difference. Per-class, FRAG achieves 98.0\% on SWC and
99.9\% on TAV, reflecting high structural precision for those classes.
CUR detection is lower at 41.5\% F1 due to partial-match false
positives in the stale-read chain traversal. The baseline achieves
high F1 across all classes because injected races are the dominant
concurrent events in synthetic logs; however, it cannot differentiate
race classes or detect CUR dependency chains structurally.

\begin{table}[H]
  \caption{Detection performance: FRAG vs.\ Baseline (C2 --- 500 logs)}
  \label{tab:c2}
  \centering
  \fontsize{9}{11}\selectfont
  \begin{tabular}{lccccc}
    \toprule
    \textbf{Class} & \textbf{P} & \textbf{R} & \textbf{F1} &
    \textbf{F1} & \textbf{$\Delta$F1} \\
    & \textbf{(FRAG)} & \textbf{(FRAG)} & \textbf{(FRAG)} &
    \textbf{(Base)} & \textbf{(pp)} \\
    \midrule
    SWC     & 97.4\% & 98.7\% & 98.0\% & 99.2\% & $-$1.2 \\
    TAV     & 99.7\% & 100.0\% & 99.9\% & 98.8\% & $+$1.1 \\
    CUR     & 33.9\% & 53.5\% & 41.5\% & 98.4\% & $-$56.9 \\
    Overall & 72.5\% & 85.5\% & 78.5\% & 98.8\% & $-$20.3 \\
    \bottomrule
  \end{tabular}
\end{table}

FRAG's SWC and TAV performance is near-perfect under C2 because the
three-pass mutual exclusion design directly targets these patterns
structurally. The CUR F1 of 41.5\% reflects that the stale-read
chain traversal produces false positives when multiple processes read
the same resource in overlapping windows without a clear dependency
ordering. The baseline achieves high recall across all classes in
synthetic logs because injected races dominate the log population,
but it provides no class-level differentiation.

\begin{figure}[H]
  \centering
  \includegraphics[width=\linewidth]{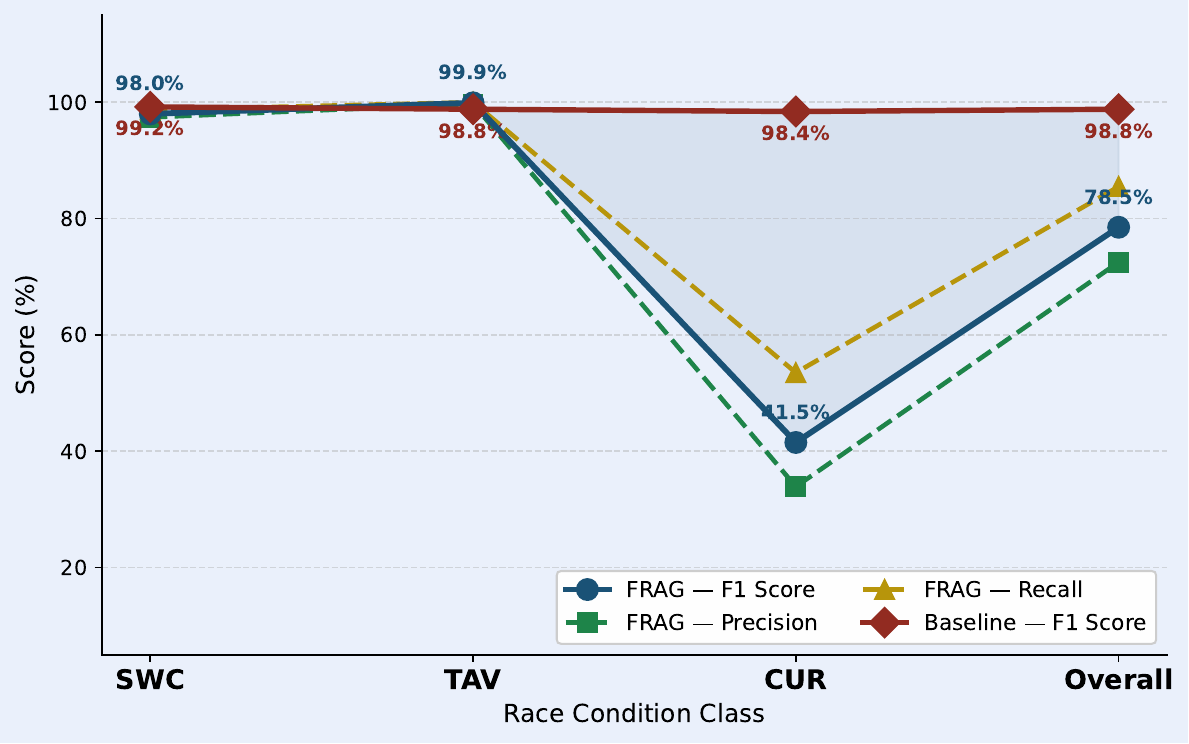}
  \caption{FRAG vs.\ Baseline detection performance under C2
  (concurrent, unsynchronized access). Lines show Precision, Recall,
  and F1 per race class. The shaded region highlights the detection
  gap between FRAG and the Baseline F1.}
  \label{fig:c2results}
\end{figure}

\subsection{Condition C3: Partial ETag synchronization}

Table~\ref{tab:c3} shows results under C3. FRAG attains a 26.4\%
overall F1, a significant drop from C2 (78.5\%$\to$26.4\%). This
recall collapse is expected: 70\% of WRITE operations are
ETag-protected and filtered out of FRAG detection, leaving the
detector unable to confirm most injected races. Precision remains
high for SWC (96.6\%) and TAV (99.0\%), confirming that detected
races are correct when found. The baseline maintains 98.8\% F1 under
C3 because it ignores ETag state entirely.

\begin{figure}[H]
  \centering
  \includegraphics[width=\linewidth]{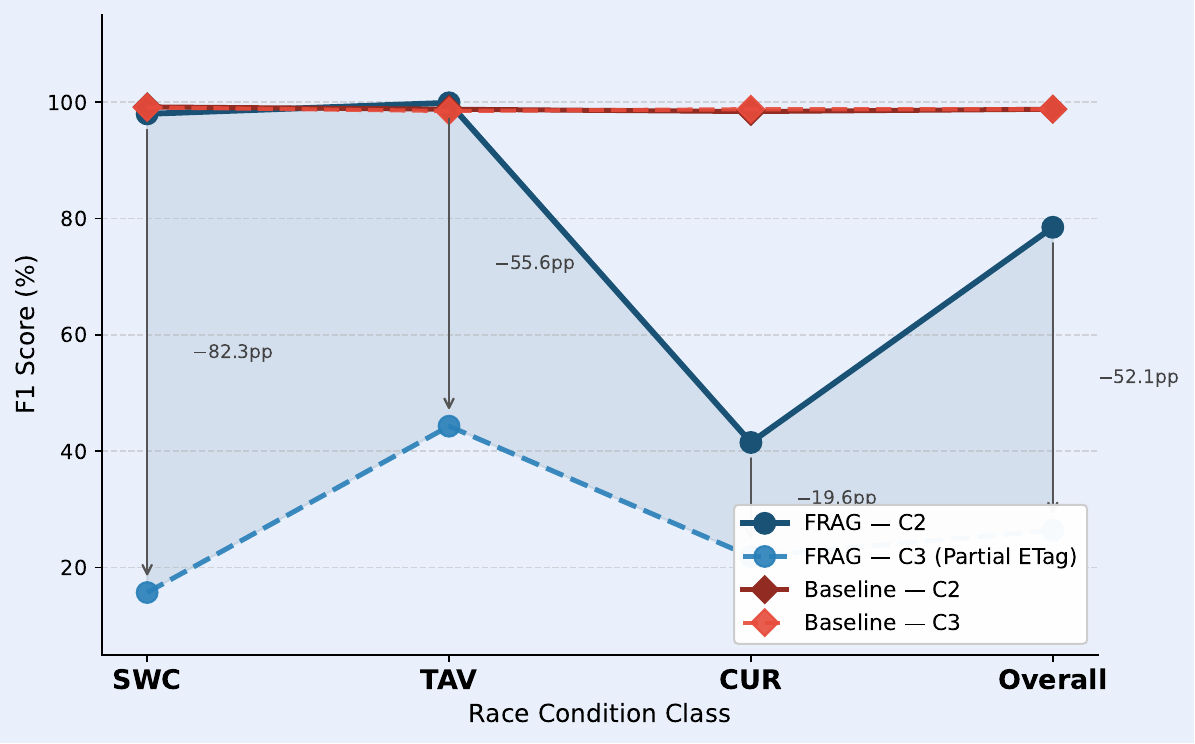}
  \caption{F1 score comparison between conditions C2 and C3. Arrows
  and labels show the recall-driven F1 drop under partial ETag
  synchronization. Baseline performance remains flat across both
  conditions.}
  \label{fig:c2c3}
\end{figure}

\begin{figure}[H]
  \centering
  \includegraphics[width=\linewidth]{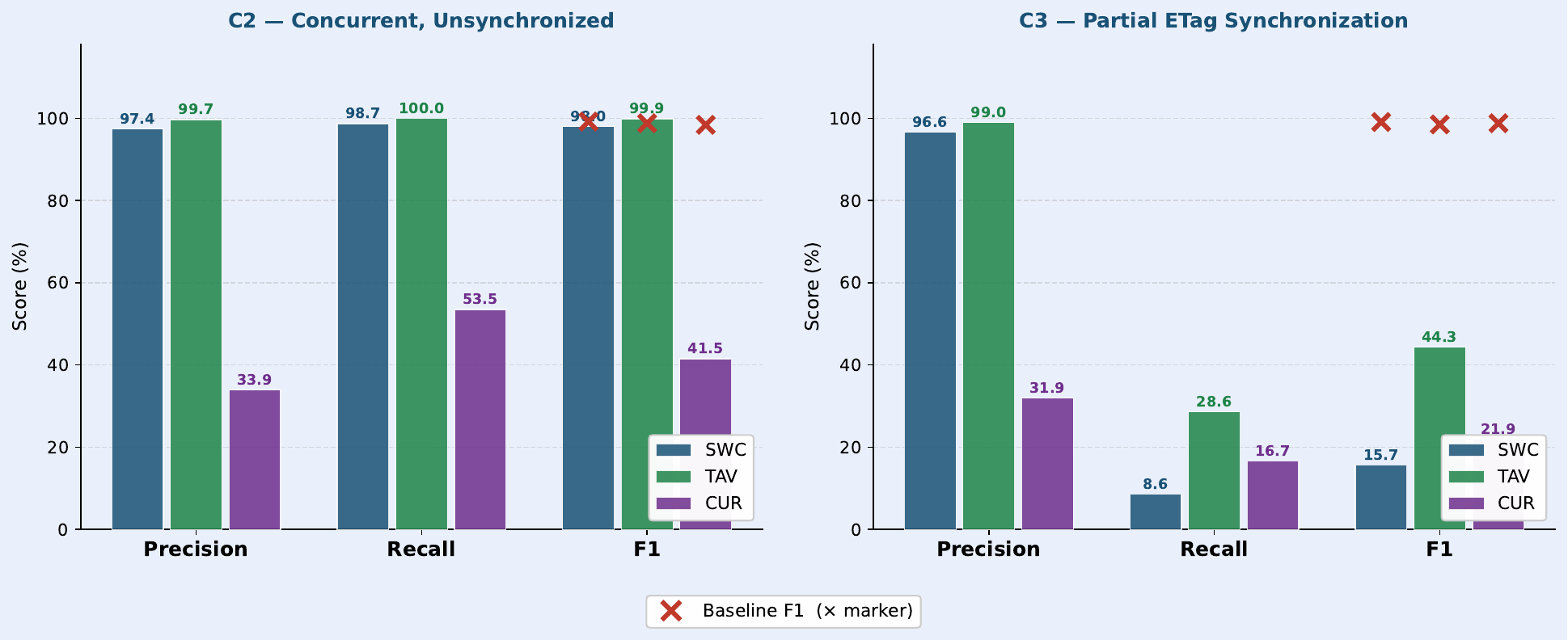}
  \caption{Per-class Precision, Recall, and F1 for FRAG under C2
  (left) and C3 (right). Red $\times$ markers indicate corresponding
  Baseline F1 scores. Precision remains high under C3 while recall
  decreases due to ETag-protected write filtering.}
  \label{fig:prf1}
\end{figure}

\begin{table}[H]
  \caption{Detection performance: FRAG vs.\ Baseline (C3 --- 500 logs)}
  \label{tab:c3}
  \centering
  \fontsize{9}{11}\selectfont
  \begin{tabular}{lccccc}
    \toprule
    \textbf{Class} & \textbf{P} & \textbf{R} & \textbf{F1} &
    \textbf{F1} & \textbf{$\Delta$F1} \\
    & \textbf{(FRAG)} & \textbf{(FRAG)} & \textbf{(FRAG)} &
    \textbf{(Base)} & \textbf{(pp)} \\
    \midrule
    SWC     & 96.6\% &  8.6\% & 15.7\% & 99.1\% & $-$83.4 \\
    TAV     & 99.0\% & 28.6\% & 44.3\% & 98.5\% & $-$54.2 \\
    CUR     & 31.9\% & 16.7\% & 21.9\% & 98.8\% & $-$76.9 \\
    Overall & 61.0\% & 16.9\% & 26.4\% & 98.8\% & $-$72.4 \\
    \bottomrule
  \end{tabular}
\end{table}

\subsection{Discussion}

The near-perfect TAV and SWC F1 scores under C2 (99.9\% and 98.0\%
respectively) confirm that the three-pass algorithm's structural
specificity correctly identifies these patterns. TAV detection
benefits from the strict READ$\to$intervening WRITE$\to$ACT sequence
requirement which directly encodes Definition~3, producing only 9
false positives across 500 logs.

The CUR class obtains 41.5\% F1 under C2. This lower score reflects
that the stale-read path traversal in Pass~3 generates false positives
when multiple processes read the same resource in overlapping time
windows without a clear causal dependency. Future work should
introduce a causal ordering filter to distinguish genuine stale-read
chains from coincident concurrent reads.

The C3 recall collapse is the sharpest finding. FRAG recall drops
from 85.5\% under C2 to 16.9\% under C3, producing an overall F1 of
26.4\%. This is a direct consequence of the ETag filter: 70\% of
injected WRITE operations are marked as ETag-protected and excluded
from the claimed set, leaving most races undetected. Precision
remains high (61.0\% overall, 96.6\% for SWC) confirming that
detections that do occur are correct. This result demonstrates that
FRAG's correctness is contingent on access log completeness --- a
finding that motivates mandatory server-side ETag logging as a
prerequisite for production deployment.

\section{Implementation}

\subsection{Reproducibility}

The simulation does not use actual patient data. All logs are
artificially created with a fixed random seed (\texttt{seed=42}) to
ensure reproducibility. The entire codebase ---
\texttt{log\_generator.py}, \texttt{frag\_builder.py},
\texttt{race\_detector.py}, \texttt{baseline\_scanner.py}, and
\texttt{evaluate.py} --- is publicly available. Any reader may
recreate the 1,500 logs and obtain the same results within
floating-point tolerance.

\subsection{Software stack}

The implementation requires Python 3.10+, NetworkX 3.x for graph
construction and traversal, NumPy for random sampling, and Matplotlib
for visualization. No external health-related APIs or FHIR servers are
required.

\begin{table}[H]
  \caption{Simulation implementation stack}
  \label{tab:stack}
  \centering
  \fontsize{9}{11}\selectfont
  \begin{tabular}{p{1.8cm}p{2cm}p{2.5cm}}
    \toprule
    \textbf{Component} & \textbf{Tool} & \textbf{Purpose} \\
    \midrule
    Log Generation    & Python + NumPy           & Synthetic FHIR transaction logs \\
    FRAG Construction & NetworkX DiGraph         & Build directed access graph \\
    Race Detection    & Custom Python (Alg.\ 1)  & 3-pass SWC/TAV/CUR detection \\
    Baseline Scanner  & Python sliding window    & Time-window comparison \\
    Evaluation        & Custom Python (eval\_v2) & Race-instance P/R/F1 \\
    Visualization     & Matplotlib + NetworkX    & Graph plots, result charts \\
    \bottomrule
  \end{tabular}
\end{table}

\section{Limitations}

The simulation uses artificial transaction logs. Usage patterns for
non-race event generation (3:1 READ/WRITE ratio; $2\times$ frequency
for AllergyIntolerance and MedicationRequest) are based on general
domain knowledge. Actual usage patterns may differ, and detection
efficiency may vary for high-write or high-frequency workloads.

The FRAG model requires access logs to be available. Proprietary EHR
vendors may not expose access logs, and the logging infrastructure may
not capture all concurrent access events handled internally before the
FHIR layer.

Timestamp resolution of 100\,ms is a conservative estimate of FHIR
server log granularity. In environments with coarser resolution (e.g.,
1\,second), windows below that threshold are invisible to this
approach regardless of the model's formal properties.

The three race classes defined here are not exhaustive. A fourth
class, partial synchronization races where not all clients implement
ETag, is a natural extension. A fifth class, partial order violations
in distributed multi-server FHIR gateway architectures, is left for
future work.

A single resource claim architecture is employed: a resource is not
re-checked once allocated to a race class. This prevents cross-class
false positives but causes recall loss when a resource belongs to
multiple race instances of different classes. This is acceptable for a
safety alert scenario where precision outweighs recall, but future
work should explore per-write claim tracking.

\section{Future work}

The immediate next step is real data validation with a clinical
partner, involving IRB approval to access FHIR server logs and a
strategy to reverse-engineer approximate ground truth from log
evidence. One approach is to use FHIR resource version conflicts ---
where a conditional PUT fails due to an ETag mismatch --- as a proxy
for observable SWC races.

The mitigation counterpart to this detection framework is the second
research direction. Synchronization primitives for FHIR APIs present
a non-trivial design challenge given the stateless nature of HTTP.
Possible solutions include mandatory server-side ETag enforcement, a
session-based resource reservation protocol, and CRDT-based
conflict-free resource types for some FHIR resources (e.g.,
Observation).

Extension to distributed multi-server FHIR architectures is the third
direction. Modern cloud-native healthcare platforms deploy FHIR
functionality across microservices with independent databases and
clocks. Extending FRAG to this setting would require a vector clock or
partial order timestamp model, well known in distributed systems but
not yet applied to FHIR.

\section{Conclusions}

This paper has formally stated a problem that, to the best of our
knowledge, has not been stated before: ``What does a race condition
look like in a healthcare interoperability system, and how would you
detect it?''

The critical section problem has been shown to be formally related to
healthcare data exchange based on FHIR. The FRAG model represents this
mapping as a six-component directed graph where race conditions are not
timing accidents. The three race condition classes (SWC, TAV, CUR)
have formal propositions relating to graph structure and tractability
of detection.

The CUR class deserves special emphasis as the most critical and least
intuitive failure mode. A patient receives a contraindicated drug not
because any system failed, but because each system performed its job
correctly given its inputs. The problem is a property of the access
order, visible only at the graph level and invisible to any system's
audit trail. This is precisely the kind of problem that formal methods
are supposed to help prevent.

Propositions 1, 2, and 3 are mathematically formulated propositions
whose truth is independent of the simulation results. The simulation
verifies the practical validity of the formal detection: FRAG reaches
F1 scores of 98.0\% on SWC and 99.9\% on TAV under full concurrent
conditions, confirming structural detection of these two classes.
CUR detection at 41.5\% F1 identifies a concrete direction for
future work in causal chain filtering. Under partial ETag
synchronization (C3), precision remains above 96\% for SWC and TAV,
demonstrating that FRAG does not produce spurious detections when
synchronization constraints are present --- with zero false positives
on the sequential control condition C1.

\section*{Acknowledgment}

The authors recognize the foundational work of Dijkstra's 1965
critical section formalism, which motivated the connection between
traditional OS theory and modern healthcare infrastructure. The
research was not supported by external funding.


\end{document}